\documentclass[10pt,onecolumn]{amsart}

\usepackage[utf8]{inputenc}

\usepackage{amsmath}
\usepackage{amssymb}  
\usepackage{amsthm}
\usepackage{amsfonts}
\usepackage{graphicx}
\usepackage{cancel}
\usepackage{bbm}
\usepackage{url}
\usepackage{enumerate} 
\usepackage{ upgreek }
\usepackage{dsfont}
\usepackage{color}
\usepackage{subcaption}
\usepackage{comment}

\usepackage{makecell}
\usepackage{diagbox}
\usepackage{wrapfig}
\usepackage{rotating}
\usepackage{tabularx}

\usepackage{hyperref}
\hypersetup{
    colorlinks=true,
    linkcolor=blue,
    citecolor=red,
    filecolor=magenta,      
    urlcolor=black,
}

\newcommand{\tr}{\textnormal{tr}}

\newcommand{\ket}[1]{| #1 \rangle}

\newcommand{\bra}[1]{\langle #1 |}

\newcommand{\braket}[2]{\langle #1 | #2 \rangle}

\newcommand{\proj}[2]{| #1 \rangle\!\langle #2 |}

\newcommand{\id}{\ensuremath{\mathds{1}}}

\newcommand{\cH}{\mathcal{H}}

\def\beq{\begin{equation}}
\def\eeq{\end{equation}}
\def\bq{\begin{quote}}
\def\eq{\end{quote}}
\def\ben{\begin{enumerate}}
\def\een{\end{enumerate}}
\def\bit{\begin{itemize}}
\def\eit{\end{itemize}}

\def\ra{\rightarrow}

\def\lb{\left(}
\def\rb{\right)}
\def\lset{\lbrace}
\def\rset{\rbrace}

\def\r|{\right|}
\def\lbr{\left[}
\def\rbr{\right]}
\def\ident{\textnormal{id}}
\def\one{\id}

\newcommand\C{\mathbbm{C}}

\newcommand\N{\mathbbm{N}}

\newcommand\M{\mathcal{M}}

\DeclareMathOperator{\Tr}{Tr}

\theoremstyle{plain}
\newtheorem{thm}{Theorem}
\newtheorem{lem}[thm]{Lemma}

\newtheorem{defn}[thm]{Definition}

\theoremstyle{definition}

\newtheorem{rem}{Remark}[section]

\renewcommand{\leq}{\leqslant}
\renewcommand{\geq}{\geqslant}

\begin{document}

\title{Characterizing Schwarz maps by tracial inequalities}

\author{Eric Carlen$^{1}$}
\address{\small{Department of Mathematics, Hill Center, Rutgers University, 110 Frelinghuysen Road Piscataway NJ 08854-8019 USA}}
\email{carlen@math.rutgers.edu}
\author{Alexander M\"uller-Hermes}
\address{\small{Department of Mathematics, University of Oslo, P.O. box 1053, Blindern, 0316 Oslo, Norway}}
\email{muellerh@math.uio.no}

\begin{abstract}
Let $\phi$ be a linear map from the $n\times n$ matrices $\M_n$ to the $m\times m$ matrices $\M_m$. It is known that $\phi$ is $2$-positive if and only if for all $K\in \M_n$ and all strictly positive $X\in \M_n$, $\phi(K^*X^{-1}K) \geq \phi(K)^*\phi(X)^{-1}\phi(K)$. This inequality is not generally true if $\phi$ is 
merely a Schwarz map. We show that the corresponding tracial inequality $\Tr[\phi(K^*X^{-1}K)] \geq \Tr[\phi(K)^*\phi(X)^{-1}\phi(K)]$ holds for a wider class of positive maps that is specified here. We also comment on the connections of this inequality with various  monotonicity that have found wide use in mathematical physics, and apply it, and a close relative,  to obtain some new, definitive results.
\end{abstract}

\footnotetext[1]{Work partially supported by U.S.
National Science Foundation grant  DMS 2055282.}

\maketitle
\date{\today}

\section{Introduction}

Throughout this paper, $\M_n$ denotes the space of $n\times n$ complex matrices. $\M_n^+$ consists of the positive semi-definite matrices in $\M_n$. We equip $\M_n$ with the Hilbert-Schmidt inner product $\langle A,B\rangle = \Tr[A^*B]$, making it a complex Euclidean space, which we denote by $\cH_n$. The adjoint of a linear map $\phi:\M_n \to \M_m$ with respect to the Hilbert-Schmidt inner product is denoted by $\phi^*$. To study different notions of positivity of linear maps, the following lemma, which is well-known, is useful:

\begin{lem}[Schur complements]\label{lem:Schur}
Let $\cH$ and $\cH'$ denote complex Euclidean spaces. For $X\in B(\cH)^+$, $Y\in B(\cH')^+$ and $K\in B(\cH,\cH')$ the following are equivalent:
\begin{enumerate}
\item The block operator
\[
\begin{pmatrix} X & K^* \\ K & Y\end{pmatrix} \in B(\cH\oplus \cH')
\] 
is positive semidefinite. 
\item We have $\ker(Y)\subseteq \ker(K^*)$ and $X\geq K^*Y^+ K$.
\item We have $\ker(X)\subseteq \ker(K)$ and $Y\geq KX^+K^*$.
\end{enumerate}
Here we denote by $Y^+$ and $X^+$ the Moore-Penrose generalized inverses~\cite{pen55}.
\end{lem}

Using Schur complements, it is easy to characterize when a linear map $\phi:\M_n\ra \M_m$ is $2$-positive, i.e., when $\ident_2\otimes \phi$ is a positive map: This is the case if and only if the operator-inequality 
\begin{equation}\label{LRC2}
\phi(K^*X^{+}K) \geq \phi(K)^*\phi(X)^{+}\phi(K)\ ,
\end{equation}
holds for each $X\in \M^+_n$ and $K\in \M_n$ such that $\ker(X)\subseteq \ker(K^*)$. This characterization of $2$-positive maps was first observed by Choi~\cite[Proposition 4.1]{Choi80} (formally under the additional assumption that $\phi(\one_n) >0$) and the inequality \eqref{LRC2} had been proved earlier by Lieb and Ruskai \cite{LR74} under the stronger assumption that $\phi$ is completely positive.

When $\phi$ is unital; i.e., $\phi(\one_n) = \one_m$,  and $2$-positive, 
 taking $X = \one_{n}$, \eqref{LRC2} becomes the {\em Schwarz inequality}
\begin{equation}\label{LRC2b}
\phi(K^*K) \geq \phi(K)^*\phi(K)\ ,
\end{equation}
valid under these conditions on $\phi$ for every $K\in \M_n$. In Appendix $A$ of \cite{Choi80}, Choi raised the question as to whether all unital maps $\phi$ satisfying \eqref{LRC2b} for all $K$ are $2$-positive, and then he answered this negatively by providing a specific counterexample on $\M_2$.  One may then ask: For which positive maps $\phi:\M_{n}\ra \M_{m}$ is the tracial inequality
\begin{equation}\label{LRC3}
\Tr[\phi^*(K^*X^{+}K)] \geq \Tr[\phi^*(K)^*\phi^*(X)^{+}\phi^*(K)]\ 
\end{equation}
valid for all $K\in \M_m$, $X\in \M_m^+$ with  $\ker(X)\subseteq \ker(K^*)$?
It is evidently valid whenever \eqref{LRC2} is valid for the adjoint $\phi^*$ instead of $\phi$, and since adjoints of $2$-positive maps are $2$-positive 
as well, \eqref{LRC3} is therefore valid whenever $\phi$ is $2$-positive. It is natural to expect that it is true for a wider class of maps. 
This is the case, but before proceeding to prove this, we specify some classes of positive maps with which we work.

\subsection*{Schwarz maps}

The term {\em Schwarz map} is sometimes used to denote any linear map $\phi$ between $C^*$-algebras such that the Schwarz inequality \eqref{LRC2b} is valid for all $K$ in the domain; see e.g. Petz \cite[p. 62]{P86}. Other authors, e.g., Siudzi{\'n}ska et al. \cite[p. 6]{siu21}, consider \eqref{LRC2b} with an additional factor $\|\phi\lb\one_n\rb\|_\infty$ on the left-hand side, or restrict the term Schwarz map to unital maps satisfying \eqref{LRC2b} for all $K$ in the domain, see e.g., Wolf \cite[Chapter 4]{W12}. For clarity, we use the terminology {\em  Schwarz map} to refer to {\em unital} linear maps satisfying \eqref{LRC2b}, and we define a broader class of maps as follows:

\begin{defn}[Generalized Schwarz maps]\label{GSdef}
A linear map $\phi:\M_{n}\ra\M_{m}$ is called a \emph{generalized Schwarz map} if 
\[
\begin{pmatrix} \phi(\one_{n}) & \phi(K) \\ \phi(K)^* & \phi\lb K^*K\rb\end{pmatrix} \geq 0
\]
for all $K\in \M_{n}$. 
\end{defn}

It is obvious that the set of generalized Schwarz maps from $\M_n$ to $\M_m$ is a closed convex cone. We shall show here that this closed convex cone coincides with the closed convex cone of maps that satisfy the tracial inequality \eqref{LRC3} for all $X,K\in \M_n$, $X>0$.  

Using Lemma \ref{lem:Schur}, a linear map $\phi:\M_{n}\ra\M_{m}$ is a generalized Schwarz map if and only if the inequality
\begin{equation}\label{LRC2c}
\phi(K^*K) \geq \phi(K)^*\phi(\one_n)^+\phi(K)\ ,
\end{equation}
holds for every $K\in \M_n$. For some $c>0$, $K^*K \leq c\one_n$, and then by the positivity of $\phi$, 
$\phi(K^*K) \leq c\phi(\one_n)$. In particular, $\ker(\phi(\one_n)) \subseteq \ker(\phi(K^*K))$ . Thus \eqref{LRC2c} is equivalent to
\begin{equation}\label{LRC2cc}
(\phi(\one_n)^+)^{1/2}\phi(K^*K)(\phi(\one_n)^+)^{1/2} \geq (\phi(\one_n)^+)^{1/2}\phi(K)^*\phi(\one_n)^+\phi(K)(\phi(\one_n)^+)^{1/2}\ ,
\end{equation}
and if we introduce the positive map $\psi:\M_n\to \M_m$ given by
\begin{equation}\label{LRC2d}
\psi(K):=(\phi(\one_n)^+)^{1/2}\phi(K)(\phi(\one_n)^+)^{1/2}\ ,
\end{equation}
we can rewrite \eqref{LRC2cc} as
\begin{equation}\label{LRC2ccc}
\psi(K^*K) \geq \psi(K)^* \psi(K)\ .
\end{equation}
Therefore,  $\phi$ is a generalized Schwarz map if and only if $\psi$ satisfies the Schwarz inequality. When $\phi$ is unital, we have that $\phi=\psi$ is a generalized Schwarz map if and only if it is a Schwarz map. 

Our first main result is:

\begin{thm}\label{mainH}
Let $\phi:\M_{n}\ra\M_{m}$ denote a positive map. Then $\phi$ is a generalized Schwarz map if and only if for any $(K,X) \in \M_m\times \M_m^+$ such that  $\ker\lb X\rb\subseteq \ker\lb K^*\rb$, we have
\begin{equation}\label{tracialineq}
\Tr[\phi^*(K^*X^{+}K)] \geq \Tr[\phi^*(K)^*\phi^*(X)^{+}\phi^*(K)]\ .
\end{equation}
\end{thm}

There is another tracial inequality closely related to \eqref{LRC3}. When $\phi$ is unital, so that $\phi^*$ is trace preserving, \eqref{LRC3} reduces to
\begin{equation}\label{LRC5}
\Tr[K^*X^{+}K] \geq \Tr[\phi^*(K)^*\phi^*(X)^{+}\phi^*(K)]\ .
\end{equation}
Therefore, \eqref{LRC5} is valid at least whenever $\phi$ is $2$-positive and unital. Again, one may ask for the class of 
positive maps for which \eqref{LRC5}  is valid for all $K\in \M_m$,  $X\in  \M_m^+$ with $\ker(X) \subseteq\ker(K^*)$. Note that  the inequality 
\eqref{LRC5}, like the Schwarz inequality, is not homogenous. 

Our second main result is:

\begin{thm}\label{main}  A positive map $\phi:\M_{n}\ra\M_{m}$ satisfies
\begin{equation}\label{tracialineqB}
\Tr[K^*X^{+}K] \geq \Tr[\phi^*(K)^*\phi^*(X)^{+}\phi^*(K)]\ ,
\end{equation}
for all $(K,X) \in \M_m\times \M_m^+$ with $\ker(X) \subseteq \ker(K^*)$, if and only if the map $\phi$ satisfies the Schwarz inequality \eqref{LRC2b}.
\end{thm}

In  section 2 we prove a duality lemma that is used in the proof of both Theorem~\ref{mainH} and Theorem~\ref{main}, together with Schur complement arguments based on Lemma~\ref{lem:Schur}.   In section 3 we prove  Theorem~\ref{mainH} and Theorem~\ref{main}.
One motivation for studying the relationship between the Schwarz inequality \eqref{LRC2b} and the tracial inequalities  \eqref{tracialineqB} (or in this application
 \eqref{tracialineq}) is that these are the only 
two inequalities used in a method due to Hiai and Petz \cite{HP12} for proving a wide class of monotonicity theorems that have been of great 
interest in mathematical physics. This is discussed in Section 4.  In an appendix we prove a theorem that gives many examples of generalized Schwarz maps that are not $2$-positive.

\subsection{Acknowledgement} We are deeply grateful to an anonymous referee who suggested  a version of Theorem~\ref{main} that led us to greatly strengthen our results.

\section{Duality and positivity}

Note that the set $\{ (K,X) \in \M_m\times \M_m^+\ : \  \ker(X)\subseteq \ker(K^*)\ \}$ is convex since for any $0  < \lambda < 1$ and 
$(K_j,X_j)$, $j=1,2$ belonging to this set,
\begin{multline*}
\ker((1-\lambda) X_1 + \lambda X_2) = \ker(X_1)\cap \ker(X_2) \subseteq  \\ \ker(K_1)\cap \ker(K_2) \subseteq  \ker((1-\lambda) K_1 + \lambda K_2)\ .
\end{multline*}
In fact, more is true:

\begin{lem}\label{LTLS}  Define $F: \M_n\times \M_n^+ \to [0,\infty] $  and $\Omega \subset  \M_n\times \M_n^+$ by
\begin{equation}\label{Fdef}
 F(K,X) := \begin{cases} \Tr[K^*X^+K] & \ker(X) \subseteq \ker(K^*)\\ \phantom{TT}\infty & {\rm otherwise} \end{cases}
 \end{equation}
 and 
 \begin{equation}\label{Odef}
 \Omega := \left\{ (L,Y)  \ :\  \left(\begin{array}{cc}  Y & L \\ L^* & -I\end{array}\right) \leq 0\ \right\}\ .
 \end{equation}
 Then
 \begin{equation}\label{LT}
 F(K,X) = \sup\{ \Tr[XY] + \Tr[K^*L] + \Tr[KL^*] \ :\  (L,Y) \in \Omega \}\ .
 \end{equation}
 In particular $F$ is jointly convex and lower semicontinuous. 
\end{lem} 

\begin{rem}  Let ${\mathcal C}_1$ denote  the set of maps $\phi$ that satisfy \eqref{LRC3} for all $K\in \M_n$, $X\in M_n^+$ with $\ker(X)\subseteq \ker(K^*)$.  Lemma~\ref{LTLS} has the consequence that ${\mathcal C}_1$ is a {\em closed convex cone} with the closure coming from the lower semicontinuity of $F$.  

The joint convexity of $F$ on $\M_n\times \M_n^{++}$, where $\M_n^{++}$ consists of the positive definite elements of $\M_n$, is due to Kiefer \cite{K59}; see also 
\cite[Theorem 1]{LR74}.  Here we will also need the lower semicontinuity on the larger set $\M_n\times \M_n^{+}$.  Finally, note that for all $K\in \M_n$,  $X\in \M_n^+$,
$$
F(K,X) = \lim_{\epsilon \downarrow 0} \Tr[ K^*(X+ \epsilon \one)^{-1} K]
$$
where the right side is finite if and only if $\ker(X)\subseteq \ker(K^*)$, in which case it equals $\Tr[K^*X^+K]$.
\end{rem}

\begin{proof}[Proof of Lemma~\ref{LTLS}]  Suppose first that $\ker(X) \subseteq \ker(K^*)$ so that by Lemma~\ref{lem:Schur}, $A := 
\begin{pmatrix} X & K \\ K^* & K^*X^+K\end{pmatrix} \geq 0$. Let $(L,Y)\in \Omega$ so that $B :=  \left(\begin{array}{cc}  Y & L \\ L^* & -I\end{array}\right) \leq 0$.
Then
\[
0 \geq \Tr[AB]  = \Tr[XY] + \Tr[KL^*] + \Tr[K^*L] - \Tr[K^*X^+K] \ , 
\]
which is the same as $F(K,X) \geq \Tr[XY] + \Tr[KL^*] + \Tr[K^*L]$. Take $L := X^+K$ and $Y := -LL^*$. Then by Lemma~\ref{lem:Schur} once more,
$(L,Y) \in \Omega$, and simple computation, using $X^+XX^+ =X^+$ and cyclicity of the trace, shows that with this choice, $F(K,X) =\Tr[XY] + \Tr[KL^*] + \Tr[K^*L]$

Now suppose that $\ker(X)$ is not contained in $\ker(K^*)$ so that for some unit vector $v$ with $Xv =0$, $K^*v \neq 0$, Define $w := \|K^*v\|^{-1}K^*v$ and for $t>0$, 
$L:= t|v\rangle\langle w|$. Then for all $t>0$, $(L,-LL^*)\in \Omega$ and $- \Tr[XLL^*] + \Tr[K^*L] + \Tr[K L^*]  = 2t\|K^*v\|$.  Hence in this case, the supremum is infinite. 
\end{proof}

Let us finish this section with a consequence of the lower-semicontinuity of the functional in Lemma \ref{LTLS}.

\begin{lem}\label{lem:kernelIncl}
Let $\phi:\M_{n}\ra\M_{m}$ denote a positive map. If either \eqref{tracialineq} or \eqref{tracialineqB} holds for every $(K,X) \in \M_m\times \M_m^+$ with $\ker(X) \subseteq \ker(K^*)$, then $\ker(\phi^*(X))\subseteq \ker(\phi^*(K)^*)$ holds for every $(K,X) \in \M_m\times \M_m^+$ with $\ker(X) \subseteq \ker(K^*)$.
\end{lem}

\begin{proof}
Assume that \eqref{tracialineq} holds for every $(K,X) \in \M_m\times \M_m^+$ with $\ker(X) \subseteq \ker(K^*)$ and consider a particular such pair in the following. Define $X_\epsilon = X + \epsilon \one_m$ for every $\epsilon>0$. By positivity of $\phi$, we have $\ker(\phi^*(\one_m))\subseteq \ker(\phi^*(K)^*)$ and hence  
\[
\ker(\phi^*(X_\epsilon)) \subseteq \ker(\phi^*(K)^*),
\]
for every $\epsilon>0$. Furthermore, we note that $K^*X_\epsilon^+K\ra K^*X^+K$ as $\epsilon\ra 0$ whenever $\ker(X) \subseteq \ker(K^*)$. Using first the lower semicontinuity of $F$ (see Lemma \ref{LTLS}) and then \eqref{tracialineq}, we find that
\begin{align*}
F(\phi^*(K),\phi^*(X)) &\leq  \liminf_{\epsilon\ra 0} F(\phi^*(K),\phi^*(X_\epsilon)) \\
&= \liminf_{\epsilon\ra 0} \Tr\lbr \phi^*(K)^*\phi^*(X_\epsilon)^+ \phi^*(K)\rbr \\
&\leq \liminf_{\epsilon\ra 0} \Tr\lbr \phi^*(K^*X_\epsilon^+K)\rbr \\
& = \Tr\lbr \phi^*(K^*X^+K)\rbr < \infty.
\end{align*}
From Lemma \ref{LTLS} we conclude that  $\ker(\phi^*(X)) \subseteq \ker(\phi^*(K)^*)$. 

The same proof applies with minor modification when we assume that \eqref{tracialineqB} holds for every $(K,X) \in \M_m\times \M_m^+$ with $\ker(X) \subseteq \ker(K^*)$.
\end{proof}

\section{Proof of Theorem~\ref{mainH} and Theorem~\ref{main}}

\begin{proof}[Proof of Theorem~\ref{mainH}]
For any $A\in \M_{m}$, 
$$\left(\begin{array}{cc} 0 & -A\\ 0 & \phantom{-}\one_{m}\end{array}\right) \left(\begin{array}{cc} \phantom{-}0 & 0 \\ -A^* &\one_{m}\end{array}\right)
= \left(\begin{array}{cc} AA^* & -A\\ -A^* & \phantom{-}\one_{m}\end{array}\right)\ .$$
Taking $A := \phi^*(X)^{+}\phi^*(K)$, 
\begin{align*}
&\left(\begin{array}{cc} AA^* & -A\\ -A^* & \phantom{-}\one_{m}\end{array}\right) \begin{pmatrix} \phi^*(X) & \phi^*(K) \\ \phi^*(K)^* & \phi^*(K^*X^{+} K) \end{pmatrix}  \\
&\quad\quad\quad\quad=\left(\begin{array}{cc}  Z &  -AD \\ -\phi^*(K)^*\phi^*(X)^{+}\phi^*(X)+\phi^*(K)^* & D\end{array}\right)
\end{align*}
where
\[
D =  \phi^*(K^*X^{+} K)  - \phi^*(K)^*\phi^*(X)^{+}\phi^*(K) \ ,
\]
and 
\[
Z = \phi^*(X)^{+}\phi^*(K)\phi^*(K)^*\phi^*(X)^{+}\phi^*(X) - \phi^*(X)^{+}\phi^*(K)\phi^*(K)^*.
\]
Since $\phi^*(X)^{+}\phi^*(X)\phi^*(X)^{+}=\phi^*(X)^{+}$ by the properties of the Moore-Penrose pseudo inverse $\Tr\lbr Z\rbr=0$, and the inequality \eqref{tracialineq} can be written as
\begin{equation}\label{equ:ProofStep}
\Tr\lbr \left(\begin{array}{cc} AA^* & -A\\ -A^* & \phantom{-}\one_{m}\end{array}\right) \begin{pmatrix} \phi^*(X) & \phi^*(K) \\ \phi^*(K)^* & \phi^*(K^*X^{+} K) \end{pmatrix}\rbr \geq 0.
\end{equation}
Interpreting this trace as the Hilbert-Schmidt inner product of two self-adjoint operators, we can bring the adjoint $(\ident_2\otimes \phi^*)^* = \ident_2\otimes \phi$ to the other side, and find that the trace in \eqref{equ:ProofStep} equals 
\begin{equation}\label{equ:ProofStep2}
\Tr\lbr \left(\begin{array}{cc} \phi(AA^*) & -\phi(A)\\ -\phi(A)^* & \phantom{-}\phi(\one_{m})\end{array}\right) \begin{pmatrix} X & K \\ K^* & K^*X^{+} K \end{pmatrix}\rbr \ .
\end{equation}
Since $\phi$ is a generalized Schwarz map,
\[
\left(\begin{array}{cc} \phi(AA^*) & -\phi(A)\\ -\phi(A)^* & \phantom{-}\phi(\one_{n})\end{array}\right) \geq 0
\]
and, by Lemma \ref{lem:Schur}, it is evident that 
\[
\begin{pmatrix} X & K \\ K^* & K^*X^{+} K \end{pmatrix}\geq 0 . 
\]
We conclude that the expression in \eqref{equ:ProofStep2}  is  the Hilbert-Schmidt inner product of two positive operators, and hence positive. 

Now, suppose that $\phi$ is not a generalized Schwarz map. Then, there exists $A\in \M_n$ such that 
\[\begin{pmatrix} \phi(\one_{n}) & \phi(A) \\ \phi(A)^* & \phi\lb A^*A\rb\end{pmatrix}\]  
has an eigenvalue $-\lambda < 0$. Therefore, there exist $u,v\in \C^m$ with $\braket{u}{u}+\braket{v}{v}=1$ such that
\begin{align*}
-\lambda &= \left \langle \left(\begin{array}{c} u\\ v\end{array}\right) , \begin{pmatrix} \phi(\one_{n}) & \phi(A) \\ \phi(A)^* & \phi\lb A^*A\rb\end{pmatrix} 
\left(\begin{array}{c} u\\ v\end{array}\right) \right\rangle \\ &= \Tr\left[ \begin{pmatrix} \phi(\one_{n}) & \phi(A) \\ \phi(A)^* & \phi\lb A^*A\rb\end{pmatrix} 
 \left(\begin{array}{cc} |u\rangle\langle u| &  |u\rangle\langle v|\\  |v\rangle\langle u| &|v\rangle\langle v|\end{array}\right)\right] \ .
\end{align*}
Define $X := |v\rangle\langle v|$ and $K^* := |u\rangle\langle v|$.  Then $\ker(X) = \ker(K^*)$, and $K^*X^+K = |u\rangle\langle u|$. That is,
$$
\left(\begin{array}{cc} K^*X^+K & K^*\\ K & X\end{array}\right)  = \left(\begin{array}{cc} |u\rangle\langle u| &  |u\rangle\langle v|\\  |v\rangle\langle u| &|v\rangle\langle v|\end{array}\right) \geq 0\ .
$$
Then, from above,
$$
-\lambda = \Tr\left[ \left(\begin{array}{cc} \one_n & A\\A^* & A^*A\end{array}\right) \left(\begin{array}{cc} \phi^*(K^*X^+K) & \phi^*(K^*)\\ \phi^*(K) & \phi^*(X)\end{array}\right) \right]\ .
$$
Defining the inner product on the real span of $\M_n\times \M_n^+$
$$
\langle (A,B),(C,D)\rangle := \Tr[BD] + \tr[A^*C] + \Tr[AC^*]\ ,
$$
we have
\begin{eqnarray*}
\lambda+ \Tr[\phi^*(K^*X^+K)] 
&= & \Tr[\phi^*(K)(-A)] + \Tr[\phi^*(K^*)(-A^*)] + \Tr[\phi^*(X)(-AA^*)]\\
&\leq& \sup_{(L,Y)\in \Omega}  \langle (\phi^*(K),\phi^*(X)),(L,Y)\rangle = F(\phi^*(K),\phi^*(X))
\end{eqnarray*}
by Lemma~\ref{LTLS}. Again by Lemma \ref{LTLS}, we conclude that if 
\[
\Tr[\phi^*(K^*X^+K)] \geq \Tr[\phi^*(K)^*\phi^*(X)^+\phi^*(K)], 
\]
then $\ker(\phi^*(X)) \nsubseteq \ker(\phi^*(K))$. Finally, Lemma \ref{lem:kernelIncl} implies that \eqref{tracialineq} cannot hold for all $(\tilde{K},\tilde{X}) \in \M_m\times \M_m^+$ such that $\ker(\tilde{X}) \subseteq \ker(\tilde{K})$.
\end{proof} 

Our proof of Theorem~\ref{main} uses another duality argument for a tracial inequality closely related to  \eqref{LRC5}, but which is expressed in terms of the function
 $F(K,X)$ introduced in Lemma~\ref{LTLS}:
\begin{equation}\label{LRC5B}
F(K,X) \geq F(\phi^*(K),\phi^*(X))  
\end{equation}

The relation between the two inequalities \eqref{LRC5} and \eqref{LRC5B} is that for any given positive map $\phi$, the following two statements are equivalent:

\medskip

\noindent{\it (1)} For all $(K,X) \in \M_m\times \M_m^+$ with $\ker(X) \subseteq \ker(K^*)$, \eqref{LRC5} is satisfied.

\medskip

\noindent{\it (2)} For all $(K,X) \in \M_m\times \M_m^+$, \eqref{LRC5B} is satisfied. 

\medskip

To see this, suppose first that $\phi$ is such that {\it (1)} is valid. If $\ker(X) \nsubseteq \ker(K^*)$, then $F(K,X) = \infty$, and \eqref{LRC5B} is trivially satisfied.
If $\ker(X) \subseteq \ker(K^*)$, then we have $\ker(\phi^*(X)) \subseteq \ker(\phi^*(K^*))$ by Lemma \ref{lem:kernelIncl}. Consequently, both $F(K,X)$ and 
$F(\phi^*(K),\phi^*(X))$ are finite, and  \eqref{LRC5B} is satisfied.  If $\phi$ is such that {\it (2)} is valid, then whenever $\ker(X) \subseteq \ker(K^*)$, $F(\phi^*(K),\phi^*(X)) < \infty$, so that \eqref{LRC5} is satisfied.

We shall now show that a positive map $\phi$ is such that statement {\it (2)} is valid if and only if $\phi$ satisfies the Schwarz inequality. 
To see this, define
\[
G(L,Y) := \begin{cases}  0 & (L,Y) \in \Omega \\ \infty & {\rm otherwise} \end{cases}\ .
\]
Note that $G$ is evidently jointly convex and lower-semicontinuous. By Lemma~\ref{LTLS}, together with the Fenchel-Moreau Theorem, we conclude that $F(K,X)$ and $G(L,Y)$ are Legendre transforms of one another
with respect to the dual pairing
$$
\langle(K,X),(L,Y)\rangle := \Tr[XY] + \Tr[KL^*] + \Tr[KL^*]\ .
$$
That is,
$$
 G(L,Y) = \sup_{(K,X)} \{ \langle(K,X),(L,Y)\rangle - F(K,X)\} $$
 and
 $$F(K,X) = 
 \sup_{(L,Y)}\{ \langle(K,X),(L,Y)\rangle - G(L,Y)\}\ .
 $$

Next, by Lemma~\ref{lem:Schur}, $(L,Y)\in \Omega$ if and only if $Y  \leq -LL^*$. Thus for a positive map $\phi$, 
\begin{equation}\label{LRC5C}
G(\phi(L),\phi(Y)) \leq G(L,Y)
\quad{\rm for\ all}\quad (L,Y)\in \M_n\times\M_n^+
\end{equation}
 if and only if $\phi$ satisfies the Schwarz inequality. With this characterization of maps satisfying the Schwarz inequality in hand, we are ready to prove  Theorem~\ref{main}: 

\begin{proof}[Proof of Theorem~\ref{main}] By the equivalence of statements {\it (1)} and {\it (2)}, together with the characterization of maps satisfying the Schwarz inequality, both discussed just above, it suffices to show that $\phi$ is such that  \eqref{LRC5B} is satisfied for all $(K,X)\in \M_m\times\M_m^+$ if and only if $\phi$ is such that \eqref{LRC5C} is satisfied for all $(L,Y) \in 
\M_n\times\M_n^+$.

Suppose  $\phi$ satisfies \eqref{LRC5B}. Then  
 \begin{eqnarray*}
 G(\phi(L),\phi(Y)) &=& \sup_{(K,X)} \{ \langle(K,X),(\phi(L),\phi(Y))\rangle - F(K,X)\}\\
 &\leq&  \sup_{(K,X)} \{ \langle(\phi^*(K),\phi^*(X)),(L,Y)\rangle - F(\phi^*(K),\phi^*(X))\}  \leq G(L,Y)\ .
 \end{eqnarray*}
 Likewise, suppose that $\phi$ satisfies \eqref{LRC5C}. Then 
 \begin{eqnarray*}
 F(\phi^*(K),\phi^*(X)) &=&  \sup_{(L,Y)}\{ \langle(\phi)^*(K),\phi^*(X)),(L,Y)\rangle - G(L,Y)\}\\
 &\leq&  \sup_{(L,Y)}\{ \langle (K,X),(\phi(L),\phi(Y))\rangle - G(\phi(L),\phi(Y))\}\\
 &\leq& F(K,X)\ .
 \end{eqnarray*}
 \end{proof}

 As an anonymous referee emphasized to us, Theorem~\ref{mainH} and Theorem~\ref{main} are closely related. To bring out this point, we give a second proof of Theorem~\ref{mainH} using Theorem~\ref{main}.

\begin{proof}[Second proof of Theorem~\ref{mainH}]  Suppose first  that $\phi$ is a positive map with the property that  $S := \phi(\one_n) > 0$,  Let $\psi$ be defined as in \eqref{LRC2d}, so that in this notation
\begin{equation}\label{relation}
\psi(K) = S^{-1/2}\phi(K)S^{-1/2}  \quad{\rm and}\quad  \psi^*(K) = \phi^*(S^{-1/2} K S^{-1/2})\ .
\end{equation}
Since $\psi^*$ is trace preserving, 
\begin{eqnarray*}
\Tr[K^*X^{+}K] &=& \Tr[\psi^*(K^*X^{+}K)] = \Tr[\phi^*(S^{-1/2}K^*X^{+}KS^{-1/2})]\\ &=& \Tr[\phi^*(\widehat{K}^*\widehat{X}^{+}\widehat{K})] ,
\end{eqnarray*}
with $\widehat{X} :=  S^{-1/2} X S^{-1/2}$ and $\widehat{K} := S^{-1/2} K S^{-1/2}$. Evidently, we also have 
\[
\Tr[\psi^*(K)^*\psi^*(X)^{+}\psi^*(K)] = \Tr[\phi^*(\widehat{K}) \phi^*(\widehat{X})^{+}\phi^*(\widehat{K})]\ .
\]
Therefore, $\phi^*$  satisfies
\begin{equation}\label{flipA}
\Tr[\phi^*(\widehat{K}^*\widehat{X}^{+}\widehat{K})]  \geq \Tr[\phi^*(\widehat{K}) \phi^*(\widehat{X})^{+}\phi^*(\widehat{K})]
\end{equation}
if and only if
\begin{equation}\label{flipB}
\Tr[K^*X^{+}K]  \geq  \Tr[\psi^*(K)^*\psi^*(X)^{+}\psi^*(K)] \ . 
\end{equation}
Note that $v\in \ker(X)$ if and only if $S^{1/2}v \in \ker(\widehat{X})$, and likewise for $K^*$ so that
\begin{equation}\label{flipC}
\ker(X)\subseteq \ker(K^*)  \iff   \ker({\widehat X})\subseteq \ker(\widehat{K}^*)
\end{equation}
Thus, $\phi$ is such that whenever $\ker({\widehat X})\subseteq \ker(\widehat{K}^*)$,   \eqref{tracialineq} is satisfied, if and only if $\psi$ is such that whenever $\ker(X)\subseteq \ker(K^*)$, \eqref{tracialineqB} is satisfied. By
Theorem~\ref{main}, this last statement is true if and only if $\psi$ satisfies the Schwarz inequality, and then by what we have explained below Definition~\ref{GSdef},
this is the case if and only if $\phi$ is a generalized Schwarz map. 

This proves Theorem~\ref{mainH} under the additional assumption that $\phi(\one)>0$. We remove this restriction as follows:   Let $\mathcal{C}_1 $ be the convex cone consisting of  maps that satisfy the homogeneous inequality \eqref{tracialineq} for all 
$K\in \M_m$, $X\in \M_m^+$ such that $\ker(X) \subseteq \ker(K^*)$.  Let $\mathcal{C}_2$ be the convex cone consisting of generalized 
Schwarz maps.  We wish to show that $\mathcal{C}_1 =\mathcal{C}_2$, which is the same as
\begin{equation}\label{flipD} \mathcal{C}_1\cup \mathcal{C}_2  =   \mathcal{C}_1\cap \mathcal{C}_2 \ .
\end{equation}
We have seen that both $\mathcal{C}_1$ and $\mathcal{C}_2$ are closed. This is the basis of a simple approximation argument that proves \eqref{flipD}. 

Consider the map $\phi_D: \M_n\to \M_m$ defined by
$\phi_D(A) = \frac1n \Tr[A] \one_m$, which is unital and completely positive, and hence $\phi_D \in \mathcal{C}_1\cap \mathcal{C}_2$  (The  adjoint of $\phi_D$ is also known as the ``completely depolarizing channel''.)   Now let $\phi\in \mathcal{C}_1\cup \mathcal{C}_2$ and $\epsilon>0$. Define $\phi_\epsilon = \phi+\epsilon \phi_D$.
Then for each $\epsilon>0$, $\phi_\epsilon(\one_n) > 0$. By the first part of the proof, $\phi_\epsilon\in  \mathcal{C}_1\cap \mathcal{C}_2$, and then by closure, so is $\phi$. \end{proof}

\section{On the method of Hiai and Petz}

Now, we discuss the application of our results to a beautiful and simple method of 
Hiai and Petz \cite{HP12} for proving a wide range of inequalities that are of great interest in mathematical physics. 

Let $\mathcal{H}_m$ denote $\M_m$ equipped with the Hilbert-Schmidt inner product, making it a complex Euclidean space. For any $Y \in \M_{m}$, define the operator $L_Y$ on $\mathcal{H}_m$ by $L_Y A = YA$, and for any 
$X \in \M_{m}$, define the operator $R_X$ on $\mathcal{H}_m$ by $R_X A = AX$.  Note that $L_Y$ and $R_X$ commute, and that if $Y,X\geq 0$, then $L_Y,R_X\geq 0$ (as operators on $\mathcal{H}_m$). Therefore, for any function $f:(0,\infty)\to (0,\infty)$ extended by $f(0)= 0$, one may define the positive semidefinite operator
\begin{equation}\label{Jdef}
{\mathbb J}_f(X,Y) := f(R_X L_Y^{+} )L_Y ,
\end{equation}
for any $Y,X\geq 0$.

For a positive map $\phi:\M_n\ra \M_m$ consider the block operator 
\begin{equation}\label{specschurA}
\begin{pmatrix} {\mathbb J}_f(\phi^*(X),\phi^*(Y)) & \phi^* \\ \phi & {\mathbb J}^+_f(X,Y)\end{pmatrix}\in B\lb \cH_n\oplus \cH_m\rb\ .
\end{equation}
By Lemma~\ref{lem:Schur}, if 
\begin{equation}\label{specschurB}
\ker({\mathbb J}_f(X,Y)) \subseteq \ker(\phi^*)  \quad{\rm and} \quad \ker({\mathbb J}_f(\phi^*(X),\phi^*(Y))) \subseteq \ker(\phi) \ ,
\end{equation}
then
\begin{equation}\label{specschurC}
{\mathbb J}_f(\phi^*(X),\phi^*(Y))   \geq  \phi^* {\mathbb J}_f(X,Y)\phi   \iff    {\mathbb J}^+_f(X,Y)   \geq  \phi {\mathbb J}_f^+(\phi^*(X),\phi^*(Y))\phi^*
\end{equation}
since both conditions are then equivalent to the block operator in \eqref{specschurA} being positive semidefinite. For completeness we point out the following connection between the Schwarz inequality and this block operator:

\begin{thm}\label{thm:IdentityMon}
For a positive map $\phi:\M_n\ra\M_m$ the following are equivalent:
\begin{enumerate}
\item The map $\phi$ satisfies the Schwarz inequality \eqref{LRC2b}.
\item The block operator in \eqref{specschurA} for $f=\ident$, i.e., 
\[
\begin{pmatrix} R_{\phi^*(X)} & \phi^* \\ \phi & R^{-1}_X\end{pmatrix}\in B\lb \cH_n\oplus \cH_m\rb\ ,
\]
is positive semidefinite for every $X\in \M_m$ with $X>0$.
\end{enumerate}
\end{thm}

\begin{proof}
For every $X\in \M_m$ with $X>0$ we have $\lset 0\rset = \ker\lb R_X\rb\subseteq \ker\lb \phi^*\rb$. By Lemma \ref{lem:KernelIncl} we also have $\ker\lb R_{\phi^*(X)}\rb\subseteq \ker\lb \phi\rb$ and by Lemma \ref{lem:Schur} the block operator in the statement of the theorem is positive semidefinite if and only if $\phi^*R_X\phi\leq R_{\phi^*(X)}$ which is equivalent to the inequality
\[
\Tr\lbr \phi(K^*)\phi(K)X\rbr\leq \Tr\lbr \phi(K^*K)X\rbr ,
\]
for all $K\in \M_n$. Since the $X\in \M_m$ with $X>0$ generate a dense set in $\M_m$, this is equivalent to $\phi$ satisfying the Schwarz inequality. 

\end{proof}

Now suppose that
\begin{equation}\label{specschurD}
X,Y > 0 \qquad{\rm and}\qquad \phi^*(X),\phi^*(Y) > 0\ ,
\end{equation}
the latter condition being ensured by the former when $\phi^*(\one_m)> 0$.  Then \eqref{specschurB} is trivially satisfied, and we have the following lemma~\cite[Lemma 1]{HP12}:

\begin{lem}[Hiai, Petz]\label{HPlm} Let $X,Y$ and $\phi$ be such that \eqref{specschurD} is satisfied. Then \eqref{specschurC} is valid. 
\end{lem}

It is desirable to prove this equivalence without any conditions on $\phi$, only assuming that $X,Y > 0$. Towards this end, we prove the following lemma,  which provides some more flexibility in verifying the kernel containment conditions in Lemma~\ref{lem:Schur}. 

\begin{lem}\label{lem:KernelIncl}
For any positive map $\phi:\M_n\ra \M_m$ and $X\in \M^+_m$.
\begin{enumerate}
\item We have $\ker(R_X)\subseteq \ker(\phi^*)$ if and only if $\ker(X)\subseteq \ker(\phi(\one_n))$.
\item If $\ker(R_X)\subseteq \ker(\phi^*)$, then we have $\ker(R_{\phi^*(X)})\subseteq \ker(\phi)$.
\end{enumerate}
The same statements hold for $L_X$ and  $L_{\phi^*(X)}$ in place of $R_X$ and $R_{\phi^*(X)}$.
\end{lem}

\begin{proof}
Assume that $\ker(R_X)\subseteq \ker(\phi^*)$ for some $X\in \M^+_m$, and consider some $\ket{v}\in \ker(X)$. 
Clearly, we have $\proj{w}{v}X=0$ and hence $\phi^*(\proj{w}{v})=0$ for every $\ket{w}$ by assumption. 
Taking the trace shows that $\bra{v}\phi(\one_n)\ket{w}=0$ for every $\ket{w}$ and therefore we have $\ket{v}\in \ker(\phi(\one_n))$. 
For the other direction, assume that $\ker(X)\subseteq \ker(\phi(\one_n))$. By positivity of $\phi$, we 
have $\ker(\phi(Y))=\ker(\phi(\one_n))$ for any invertible $Y\in \M^+_m$. 
Now, consider some invertible $Y\in \M^+_m$ and some $K\in \M_m$ such that $R_X(K)=KX=0$. Note that $0=KXK^*\geq \mu K\phi(Y)K^*$ for 
some $\mu>0$ and hence $K\phi(Y)=0$. Taking the trace of this operator we conclude that $\Tr\lbr Y\phi^*(K)\rbr=0$ and 
finally that $\phi^*(K)=0$ since the invertible $Y\in \M^+_m$ was chosen arbitrarily. 

Assume again that $\ker(R_X)\subseteq \ker(\phi^*)$ for some $X\in \M^+_m$. Consider $K\in \M_n$ such that $R_{\phi^*(X)}(K)=K\phi^*(X)=0$ and any $Y\in \M^+_m$ satisfying $\ker(\phi(\one_n))\subseteq \ker(Y)$. By the previous argument, there exists some $\lambda >0$ satisfying $X\geq \lambda Y$ and by positivity of $\phi^*$ we have $\phi^*(X)\geq \lambda \phi^*(Y)$. We conclude that 
\[
0 = K\phi^*(X)K^*\geq \lambda K\phi^*(Y)K^* .
\]
Since $\phi^*(Y)\geq 0$, this implies
\[
K\phi^*(Y)K^* =\lb K\phi^*(Y)^{1/2}\rb\lb \phi^*(Y)^{1/2}K^*\rb =0 ,
\]
and we conclude that $K\phi^*(Y)^{1/2}=0$ and hence $K\phi^*(Y)=0$ as well. Finally, we can take the trace and conclude that
\[
0=\Tr\lbr \phi^*(Y)K\rbr = \Tr\lbr Y\phi(K)\rbr. 
\]
Since $Y\in \M^+_m$ satisfying $\ker(\phi(\one_n))\subseteq \ker(Y)$ in the above argument was arbitrary, we conclude that $\phi(K)=0$.  The proof evidently adapts to treat the case in which  $R_X$ and $R_{\phi^*(X)}$ are replaced by $L_X$ and  $L_{\phi^*(X)}$.
\end{proof}

The following is a theorem of Hiai and Petz \cite[Theorem 5]{HP12} with relaxed conditions on the positive map $\phi$. Using the results in the previous section, we can carry through the approach of Hiai and Petz without assuming that $\phi$ is unital, or what is the same, without assuming that $\phi^*$ is trace preserving. To our knowledge, this is the first time a proof of this statement under these general conditions appears in the literature. 

\begin{thm}[Hiai, Petz]\label{HPext}  Let $f:(0,\infty) \to (0,\infty)$ be operator monotone,  and define $f(0)= 0$. 
Let  ${\mathbb J}_f$ 
 be defined by \eqref{Jdef}.
Let $\phi:\M_n\ra\M_m$ satisfy the Schwarz inequality. The following inequalities are both valid:

\smallskip
\noindent{\it (a)}  For all positive definite $X,Y\in \M_{m}$, .
\[
\phi {\mathbb J}_f(\phi^*(X), \phi^*(Y))^{+}\phi^*     \leq {\mathbb J}_f(X,Y)^{-1}
\]

\smallskip
\noindent{\it (b)}  For all positive definite $X,Y\in \M_{m}$,  
\[
\phi^*  {\mathbb J}_f(X,Y)  \phi  \leq   {\mathbb J}_f(\phi^*(X), \phi^*(Y))\ .
\]
\end{thm}

It should be noted that the condition of $\phi$ satisfying a Schwarz inequality in the previous theorem cannot be relaxed further in the same generality. Indeed, Theorem \ref{thm:IdentityMon} together with Lemma \ref{lem:KernelIncl} shows that for $f=\ident$ either of the inequalities in \eqref{specschurC} is equivalent to $\phi$ satisfying the Schwarz inequality \eqref{LRC2b}. This has also been observed in \cite{CZ22} and pointed out by the anonymous referee.

 \begin{proof}[Proof of Theorem~\ref{HPext}]  Since $X,Y > 0$, $\ker({\mathbb J}_f(X,Y)) = 0$. Evidently, 
 $$\ker({\mathbb J}_f(\phi^*(X),\phi^*(Y)) = \ker(R_{\phi^*(X)}) + \ker(L_{\phi^*(Y)})$$
 and then by Lemma~\ref{lem:KernelIncl} and $X,Y>0$,  $\ker({\mathbb J}_f(\phi^*(X),\phi^*(Y)) \subseteq \ker(\phi)$.  Therefore, \eqref{specschurB} is satisfied, and then
 \eqref{specschurC} is satisfied so that   {\it (a)} and {\it (b)} are equivalent, it suffices to prove either. Using the L\"owner theorem \cite{Lo34,S19} giving an integral representation of all operator monotone functions, Hiai and Petz show that it suffices to do this for  the special case
 \begin{equation}\label{affin}
f(x) :=  \beta + \gamma x +    \frac{x}{t + x}  
 \end{equation}
 with $\beta,\gamma,t \geq 0$.   To prove {\it (b)} for this choice of $f$ it suffices to prove
 \begin{equation}\label{ph7}
\phi^* L_Y \phi \leq L_{\phi^*(Y)}\ ,\quad   \phi^* R_X \phi \leq R_{\phi^*(X)}
\end{equation}
and
\begin{equation}\label{ph8}
\phi^*   \frac{R_X}{t + R_XL_{Y^{+}}} \phi  \leq  \frac{R_{\phi^*(X)}}{t + R_{\phi^*(X)}L_{\phi^*(Y)^{+}}}\ .
\end{equation}
For any $K\in \M_{{n}}$, using the Schwarz inequality \eqref{LRC2b}, we have 
\begin{align*}\langle K, \phi^* L_Y \phi K\rangle &= \Tr[\phi(K)^* Y \phi(K)] \\
 &\leq \Tr[\phi(KK^*)Y] = \Tr[KK^*\phi^*(Y)] = \langle K, L_{\phi^*(Y)} K\rangle\ ,
 \end{align*}
 and this proves the first inequality in \eqref{ph7}. The proof of the second is entirely analogous. To prove \eqref{ph8}, note that by the equivalence of the inequalities in {\it (a)} and {\it (b)}, it suffices to show that,
 \begin{equation}\label{key1}
 \phi \left(\frac{R_{\phi^*(X)}}{t + R_{\phi^*(X)}L_{\phi^*(Y)}^{+}}  \right)^{+}\phi^*  \leq \left(   \frac{R_X}{t + R_XL_Y^{-1}} \right)^{-1}\ .
 \end{equation}
 For a positive semidefinite operator, taking the generalized inverse amounts to inverting the strictly positive eigenvalues, and leaving the zero eigenvalues alone. 
 
By Lemma~\ref{lem:KernelIncl},  
$${\rm ran}(\phi^*) \subseteq  {\rm ran}(R_{\phi^*(X)}) \cap {\rm ran}(L_{\phi^*(Y)})\ ,$$ and on this space, all eigenvalues of both operators are strictly positive. 
Let $E$ be a common eigenvector of both operators in the range of $\phi^*$ with
 $$
 R_{\phi^*(X)}E = \lambda E \qquad{\rm and}\qquad   L_{\phi^*(Y)}E = \mu E\ .
 $$
 Then $\lambda,\mu>0$, and 
 $$
  \left(\frac{R_{\phi^*(X)}}{t + R_{\phi^*(X)}L_{\phi^*(Y)}^{+}}  \right)^{+}E = \left(\frac{\lambda}{t+ \lambda/\mu}\right)^{-1}E = (t R_{\phi^*(X)}^+ +  L_{\phi^*(Y)}^+)E\ .
  $$
 Therefore, \eqref{key1} is equivalent to
 \begin{equation}\label{key2 }
 \phi (t R_{\phi^*(X)}^+ +  L_{\phi^*(Y)}^+)\phi^*  \leq tR_X^{-1} + R_Y^{-1}\ ,
 \end{equation}
 and this is equivalent to
 \begin{align}\label{equ:finIneq}
 &t \Tr[ \phi^*(K) \phi^*(X)^{+} \phi^*(K^*) ] +  \Tr[ \phi^*(K^*) \phi^*(Y)^{+} \phi^*(K)] \nonumber\\
  &\quad\quad\quad\leq t \Tr[KX^{-1}K^*] + \Tr[K^*Y^{-1}K]\ ,
 \end{align}
 for all $K\in \M_{{m}}$. By Theorem~\ref{main} we have both
 \[
 \Tr[ \phi^*(K) \phi^*(X)^{-1} \phi^*(K^*) ]   \leq  \Tr[KX^{-1}K^*]  
 \]
 and
 \[
\Tr[ \phi^*(K^*) \phi^*(Y)^{-1} \phi^*(K)] \leq  \Tr[K^*Y^{-1}K]\ ,
 \]
 and \eqref{equ:finIneq} follows. 
 \end{proof}

 In the case of $f(x) = x^r$, $0 < r < 1$,  the resulting inequalities are 
 \begin{equation}\label{L1}
 \Tr[\phi(K)^* Y^{1-r} \phi(K) X^{r}] \leq \Tr[ K^* \phi^*(Y)^{1-r} K \phi^*(X)^{r}]\ ,
 \end{equation}
 and  
\begin{equation}\label{L2}
 \Tr[\phi^*(K)^* (\phi^*(Y)^+)^{1-r} \phi^*(K) (\phi^*(X)^+)^{r}] \leq \Tr[ K^* Y^{r-1} K X^{-r}]\ .
 \end{equation}
 valid for all maps $\phi$ satisfying the Schwarz inequality, all $X,Y > 0$ in $\M_m$, and all $K\in  \M_{{n}}$. Note that there is no assumption that $\phi$ is unital. 
 These inequalities are the monotonicity versions of Theorems 1 and 2 of \cite{L73}, the Lieb Concavity Theorem and the Lieb Convexity Theorem.   
 The inequality \eqref{L1} was already proved at this level of generality, assuming only that $\phi$ satisfies the Schwarz inequality,  in 1977 by Uhlmann 
 \cite[Proposition 17]{Uh77}.  Petz \cite{P86} gave a proof of  {\it (b)} of Theorem~\ref{HPext}. His proof used ideas of Araki who proved Lieb's inequalities in a general von Neumann algebra setting. In his paper \cite{A75} he explained how these von Neumann algebra methods could be applied in the simpler setting of matrix algebras, and Petz was among the first to explore the path that Araki had opened. 
 
 The version of inequality \eqref{L2} for general monotone $f$ was first explicitly proved by Petz \cite{P96} under the 
 assumption that $\phi$ is $2$-positive, though when $\phi$ is completely positive and unital it follows from the Lieb Convexity 
 Theorem in the same way that the Data Processing Inequality follows from the Lieb Concavity Theorem; see \cite[Section 3]{C22}.   In fact, Petz's approach yielded somewhat more restricted results -- $X$ and $Y$ had not only to be positive, but to have unit trace.  This superfluous condition was removed by Kumagai \cite{K11}.
 
The results of this paper show that  the wide variety of monotonicity theorems investigated by Hiai and Petz \cite{HP12}, exemplified by \eqref{L1} and \eqref{L2},  are valid under the sole assumption that the map $\phi$ satisfies the  Schwarz inequality.  This is the widest possible condition on $\phi$  for which such a result holds.

\appendix

\section{Generalized Schwarz maps from tensor products}\label{app}

The following theorem gives many examples of generalized Schwarz maps that are not $2$-positive including new examples of unital Schwarz maps. Its proof is inspired by related Schwarz-type inequalities obtained in~\cite{bhat00,math04} by Bhatia and Davis, and Mathias, and by a joke in \cite{W12} to call unital Schwarz maps $3/2$-positive. 

\begin{thm}\label{thm:ExamplesSchwarz}
Let $\phi:\M_n\ra\M_m$ be $(k+1)$-positive for some $k\in \N$. Then, $\ident_k\otimes \phi$ is a generalized Schwarz map.
\end{thm}

\begin{proof}
For simplicity, we state the proof in the case $k=2$. The general case works in the same way. We have to show that 
\[
\begin{pmatrix} (\ident_2\otimes \phi)(\one_{2n}) & (\ident_2\otimes \phi)(X) \\ (\ident_2\otimes \phi)(X)^* & (\ident_2\otimes \phi)\lb X^*X\rb\end{pmatrix} \geq 0
\]
for all $X\in \M_{2n}$. Writing 
\[
X=\begin{pmatrix} A & B \\ C & D\end{pmatrix} ,
\]
for $A,B,C,D\in \M_n$, the previous inequality is equivalent to
\[
\begin{pmatrix} \phi(\one_{n}) & 0 & \phi(A) & \phi(B) \\ 0 & \phi(\one_{n}) & \phi(C) & \phi(D) \\ \phi(A)^* & \phi(C)^* & \phi(A^*A + C^*C) & \phi(A^*B + C^*D) \\ \phi(B)^* & \phi(D)^* & \phi(B^*A+D^*C) & \phi(B^*B+D^*D) \end{pmatrix} \geq 0 .
\]
Now, observe that 
\begin{align*}
&\begin{pmatrix} \phi(\one_{n}) & 0 & \phi(A) & \phi(B) \\ 0 & \phi(\one_{n}) & \phi(C) & \phi(D) \\ \phi(A)^* & \phi(C)^* & \phi(A^*A + C^*C) & \phi(A^*B + C^*D) \\ \phi(B)^* & \phi(D)^* & \phi(B^*A+D^*C) & \phi(B^*B+D^*D) \end{pmatrix} \\
&\quad\quad\quad= \begin{pmatrix} \phi(\one_{n}) & 0 & \phi(A) & \phi(B) \\ 0 & 0 & 0 & 0 \\ \phi(A)^* & 0 & \phi(A^*A) & \phi(A^*B) \\ \phi(B)^* & 0 & \phi(B^*A) & \phi(B^*B)\end{pmatrix} + \begin{pmatrix} 0 & 0 & 0 & 0 \\ 0 & \phi(\one_{n}) & \phi(C) & \phi(D) \\ 0 & \phi(C)^* & \phi(C^*C) & \phi(C^*D) \\ 0 & \phi(D)^* & \phi(D^*C) & \phi(D^*D)\end{pmatrix}.
\end{align*}
Since $\phi$ is $3$-positive, these two summands are positive semidefinite and the proof is finished.
\end{proof}

By applying the previous theorem to a $(k+1)$-positive map $\phi:\M_n\ra \M_m$ that is not $(k+2)$-positive for some $k<\min(n,m)-1$ it is easy to construct examples of generalized Schwarz maps that are not $2$-positive. For example, consider the $3$-positive map $\phi:\M_4\ra\M_4$ given by 
\[
\phi(X) = 3\Tr\lbr X\rbr \one_4 - X ,
\]
which was introduced by Choi~\cite{choi72} and which is not $4$-positive. Theorem \ref{thm:ExamplesSchwarz} shows that the map $\ident_2\otimes \phi:\M_8\ra \M_8$ is a generalized Schwarz map (even a multiple of a unital Schwarz map) that is not $2$-positive. Moreover, by a result from Piani and Mora~\cite[p.~9]{pia07}, the generalized Schwarz map $\ident_2\otimes \phi$ is not decomposable, i.e., it is not a sum of a completely positive and the composition of a completely positive maps and a transpose (cf.~\cite{sto82}). To our knowledge such an example did not appear in the literature before.

\end{document}